\documentclass[journal,onecolumn,11pt]{IEEEtran} 
\newcommand{\url}[1]{#1}
\usepackage{graphicx}
\usepackage{amssymb,amsopn,amsmath,amsthm}
\usepackage{subfigure}

\newcommand{\eps}{\varepsilon}
\newcommand{\cP}{\mathcal{P}}
\newcommand{\cX}{\mathcal{X}}
\newcommand{\cY}{\mathcal{Y}}
\newcommand{\cZ}{\mathcal{Z}}
\newcommand{\cM}{\mathcal{M}}

\newtheorem{lemma}{Lemma}

\newtheorem{proposition}[lemma]{Proposition}
\newtheorem{corollary}[lemma]{Corollary}
\newtheorem{theorem}[lemma]{Theorem}

\DeclareMathOperator{\Exp}{\mathbb{E}}
\DeclareMathOperator{\Var}{Var}
\newcommand{\lEC}{\textnormal{leak}_{EC}}

\DeclareMathOperator{\diff}{d}

\newcommand{\href}[2]{#2}

\begin{document}

\title{Fundamental Finite Key Limits for One-Way Information Reconciliation in Quantum Key Distribution}

\author{Marco Tomamichel$^\dagger$ \and Jesus Martinez-Mateo \and  Christoph Pacher \and David Elkouss
\thanks{Marco Tomamichel is with Centre for Quantum Software and Information, University of Technology Sydney, NSW 2007, Australia. 
Jes\'us Mart\'inez-Mateo is with Center for Computational Simulation, Universidad Politecnica de Madrid, 28660 Boadilla del Monte, Spain. 
Christoph Pacher is with Digital Safety \& Security Department, AIT Austrian Institute of Technology, Donau-City-Stra{\ss}e 1, 1220 Vienna, Austria.
David Elkouss is with  QuTech, Delft University of Technology, P.O. Box 5046, 2600 GA Delft, The Netherlands. }
\thanks{Part of these results without the technical derivations were published in the proceedings of the International Symposium on Information Theory,  Honolulu (2014). \cite{Tomamichel_14}}
   \thanks{$^\dagger$\url{marco.tomamichel@sydney.edu.au}   }
   }
 \maketitle

\begin{abstract}
The security of quantum key distribution protocols is guaranteed by the laws of quantum mechanics. However, a precise analysis of the security properties requires tools from both classical cryptography and information theory. Here, we employ recent results in non-asymptotic classical information theory to show that one-way information reconciliation imposes fundamental limitations on the amount of secret key that can be extracted in the finite key regime. In particular, we find that an often used approximation for the information leakage during information reconciliation is not generally valid. We propose an improved approximation that takes into account finite key effects and numerically test it against codes for two probability distributions, that we call binary-binary and binary-Gaussian, that typically appear in quantum key distribution protocols.
\end{abstract}

\section{Introduction}
\label{intro}
Quantum key distribution (QKD)~\cite{bb84,ekert91} is a prime example of the interdisciplinary nature of quantum cryptography and the first application of quantum science that has matured into the realm of engineering and commercial development. While the security of the generated key is intuitively guaranteed by the laws of quantum mechanics, a precise analysis of the security requires tools from both classical cryptography and information theory (see~\cite{mayers01,SP00} for early security proofs, and see~\cite{scarani09a} for a comprehensive review). This is particularly relevant when investigating the security of QKD in a practical setting where the resources available to the honest parties are finite and the security analysis consequently relies on non-asymptotic information theory.

In the following, we consider QKD protocols between two honest parties, Alice and Bob, which can be partitioned into the following rough steps. In the \emph{quantum phase}, $N$ physical systems are prepared, exchanged and measured by Alice and Bob. In the \emph{parameter estimation (PE) phase}, relevant parameters describing the channel between Alice and Bob are estimated from correlations measured in the quantum phase. If the estimated parameters do not allow extraction of a secure key, the protocol aborts at this point.
Otherwise, the remaining measurement data is condensed into two highly correlated bit strings of length $n$ in the \emph{sifting phase}\,---\,the \emph{raw keys} $X^n$ for Alice and $Y^n$ for Bob~\cite{Pfister_15}. 
We call $n$ the block length and it is the quantity that is usually limited by practical considerations (time interval between generated keys, amount of key that has to be discarded in case Alice and Bob create different keys, hardware restrictions). In the \emph{information reconciliation (IR) phase}, Alice and Bob exchange classical information about $X^n$ over a public channel in order for Bob to compute an estimate $\hat{X}^n$ of $X^n$. The \emph{confirmation (CO) phase} ensures that $\hat{X}^n=X^n$ holds with high probability, or it aborts the protocol. Finally, in the \emph{privacy amplification (PA) phase}, Alice and Bob distill a shared secret key of $\ell$ bits from $X^n$ and $\hat{X}^n$. We say that a protocol is \emph{secure} if (up to some error tolerance) both Alice and Bob hold an identical, uniform key that is independent of the information gathered by an eavesdropper during the protocol, for any eavesdropper with access to the quantum and the authenticated classical channel. 

The ratio $\ell/N$ is constrained by the following effects:
1) Some measurement results are published for PE and subsequently discarded. 
2) The sifting phase removes data that is not expected to be highly correlated, thus further reducing the length $n$ of the raw key. 
3) Additional information about the raw keys is leaked to the eavesdropper during the IR and CO phase.
4) To remove correlations with the eavesdropper, $X^n$ and $\hat{X}^n$ need to be purged in the PA phase, resulting in a shorter key. 

Some of these contributions vanish asymptotically for large $N$ while others approach fundamental limits.
\footnote{Consider, for example, BB84 with asymmetric basis choice~\cite{lo04} on a channel with quantum bit error rate $Q$. There, contributions 1) and 2) vanish asymptotically while contributions 3) and 4) converge to $h(Q)$.}

Modern tools allow to analyze QKD protocols that are secure against the most general attacks. They provide lower bounds on the number of secure key bits that can be extracted for a fixed block length,~$n$. For the BB84 protocol, such proofs are for example given in~\cite{scarani08,renner05} and~\cite{hayashi06-2}. These proofs were subsequently simplified to achieve better key rates in~\cite{tomamichellim11} and~\cite{hayashi11}, respectively. (See also~\cite{Tomamichel_15} for a recent detailed proof.) All results have in common that the key rate that can be achieved with finite resources is strictly smaller than the asymptotic limit for large $n$\,---\,as one would intuitively expect. 

We are concerned with a complementary question: Given a secure but otherwise arbitrary QKD protocol for a fixed $n$, are there fundamental upper bounds on the length of the key that can be produced by this protocol? Such bounds are of theoretical as well as practical interest since they provide a benchmark against which contemporary implementations of QKD can be measured. In the asymptotic regime of large block lengths, such upper bounds have already been investigated, for example in~\cite{moroder06}. Here we limit the discussion to IR and focus on bounds that solely arise due to finite block lengths (Sec.~\ref{sec:flimits}). We complement the bounds with a numerical study of achievable leak values with LDPC codes (Sec.~\ref{sec:results}), and study some possible improvements and open issues (Sec.~\ref{sec:conclusion}).

\section{Fundamental limits for one-way reconciliation}
\label{sec:flimits}

We consider \emph{one-way} IR protocols, where Alice first computes a syndrome, $M \in \cM$, from her raw key, $X^n$, and sends it to Bob who uses the syndrome together with his own raw key, $Y^n$, to construct an estimate $\hat{X}^n$ of $X^n$. We will assume that $X$ takes values in a discrete alphabet while we allow $Y$ to take values in the real line. We are interested in the size of the syndrome (in bits), denoted $\log |\cM|$, and the probability of error, $\Pr[X^n \neq \hat{X}^n]$. In most contemporary security proofs $\log |\cM|$ enters the calculation of the key rate rather directly.
\footnote{Recent works analyzing the finite block length behavior using this approximation include~\cite{scarani08,cai09,tomamichellim11,hayashi11,bratzik10,abruzzo11,lim12}.} More precisely, to achieve security it is necessary (but not sufficient) that 

\begin{align}
  \ell \leq n - \lEC \label{eq:ell},
\end{align}

where $\lEC$ is the amount of information leaked to the eavesdropper during IR. Since it is usually impossible to determine $\lEC$ precisely, this term is often bounded as $\lEC \leq \log |\cM|$.
In the following, we are thus interested in finding lower bounds on $\log |\cM|$.

Let $f_{XY}$ be a probability density function. We say that an IR protocol is \emph{$\eps$-correct} on $f_{XY}$ if it satisfies $\Pr[X^n \neq \hat{X}^n] \leq \eps$ when $X^n$ and $Y^n$ are distributed according to $(f_{XY})^{\times n}$. Any such protocol (under weak conditions on $f_{XY}$ and for small $\eps$) satisfies $\frac{1}{n} \log |\cM| \geq H(X|Y)_f$~\cite{tan12}. Moreover, equality can be achieved for $n \to \infty$~\cite{slepian73}. On first sight, it thus appears reasonable to compare the performance of a finite block length protocol by comparing $\log |\cM|$ with its asymptotic limit. In fact, for the purpose of numerical simulations, the amount of one-way communication from Alice to Bob required to perform IR is usually approximated as $\lEC \approx \xi \cdot n H(X|Y)_f$, where $\xi > 1$ is the reconciliation efficiency. The constant $\xi$ is often chosen in the range $\xi = 1.05$ to $\xi = 1.2$. However, this choice is scarcely motivated and independent of the block length, the bit error rate and the required correctness considered.

\begin{figure}
\centering
\includegraphics[width=0.6\textwidth]{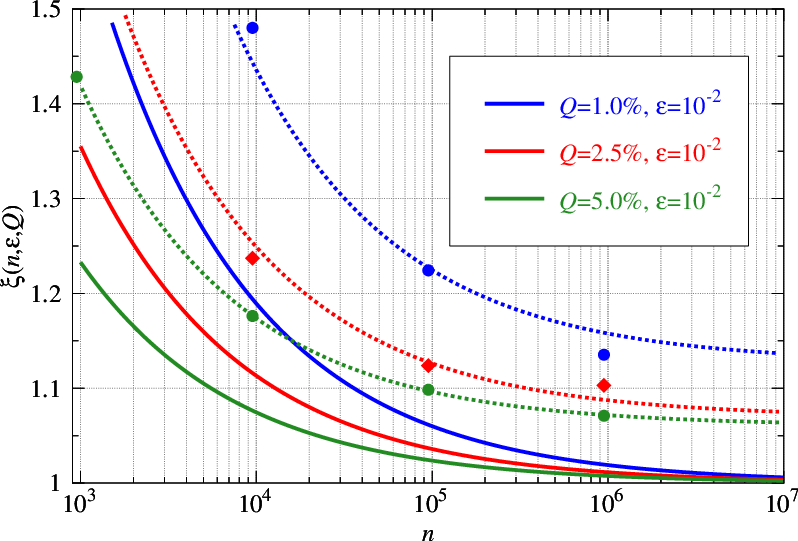}
\includegraphics[width=0.6\textwidth]{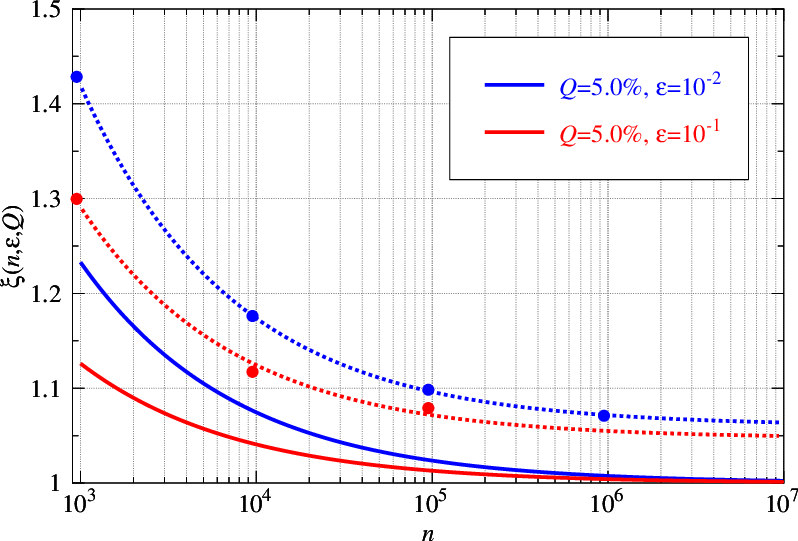}
\caption{The solid lines show the fundamental limit of the efficiency for the binary-binary distribution, $\xi(n, \eps; Q)$, as a function of $n$ for different values of $Q$ and $\eps$. The dotted lines show fits (see Table \ref{tab:xi}) to Eq.~\eqref{eq:fit} for simulated LDPC codes (marked with symbols).}
\label{fig:xi}
\end{figure}

\begin{figure}
\centering
\includegraphics[width=0.6\textwidth]{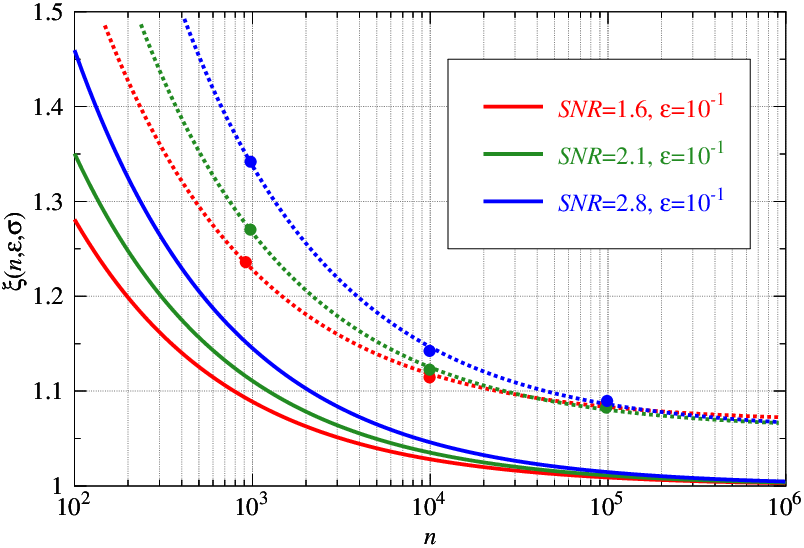}
\caption{As in Fig.~\ref{fig:xi} the solid lines show the fundamental limit of the efficiency but for the binary-Gaussian distribution, $\xi(n, \eps; \sigma)$, as a function of $n$ for different signal-to-noise ratios (SNR) and $\eps$ values.
}
\label{fig:xi-awgn}
\end{figure}

Here, we argue that this approximation is unnecessarily rough in light of recent progress in non-asymptotic information theory. Strassen~\cite{strassen62} already observed in the context of noisy channel coding that the asymptotic expansion of the fundamental limit for large $n$ admits a Gaussian approximation. This approximation was recently refined by Polyanskiy \emph{et al.}~\cite{polyanskiy10} (see also~\cite{hayashi09}). The problem of information reconciliation\,---\,also called source compression with side information\,---\,was investigated by Hayashi~\cite{hayashi08} and recently by Tan and Kosut~\cite{tan12}. Here we go slightly beyond this and provide bounds on the asymptotic expansion up to third order:

\begin{theorem}
\label{th:gen}
Let $0 < \eps <\! 1$ and $f_{XY}\!$ arbitrary. Then, for large $n$, any $\eps$-correct IR protocol on $f_{XY}$ satisfies
\begin{align*}
	\log |\cM| &\geq n H(X|Y) + \sqrt{n V(X|Y)}\, \Phi^{-1}(1-\eps) - \frac12 \log n - O(1) \,.
\end{align*}

Furthermore, there exists an $\eps$-correct IR protocol with
\begin{align*}
	\log |\cM| &\leq n H(X|Y) + \sqrt{n V(X|Y)}\, \Phi^{-1}(1-\eps) + \frac12 \log n + O(1),
\end{align*}
where $\Phi$ is the cumulative standard normal distribution, 
\begin{equation}
	H(X|Y) := \Exp \left[-\log \frac{f_{XY}}{f_{Y}} \right]
\end{equation}
is the conditional entropy and 
\begin{equation}
\label{eq:condentvar}
	V(X|Y) := \Var \left[ -\log \frac{f_{XY}}{f_{Y}}\right]
\end{equation}
is the conditional entropy variance.
\end{theorem}

The proof uses standard techniques, namely Yassaee~\emph{et al.}'s achievability bounds~\cite{yassaee13} and an analogue of the meta-converse~\cite{polyanskiy10}. 
Note that the gap of $\log n$ between achievable and converse bounds for general distributions leaves room for improvements. In channel coding, the gap is at most $\frac12 \log n$, and constant for certain channels (see, e.g.,~\cite{tomamicheltan12,altug13,tantomamichel13} for recent work on this topic).

We are in particular interested in two situations that typically appear in QKD.

\subsection{Binary Variable QKD}

We first look at binary variable protocols, such as BB84~\cite{bb84} or the $6$-state protocol~\cite{bruss98}, in the absence of an active eavesdropper. In this situation, the raw keys $X$ and $Y$ result from measurements on a channel with independent quantum bit error rate $Q$. The distribution $(P_{XY}^Q)^n$, that we call the binary-binary distribution, describes a typical manifestation of two random strings for which the expected bit error rate is $Q$. Here, we (at least) require $\eps$-correctness for the distribution
\begin{align}
P_{XY}^Q(0,0) &= P_{XY}^Q(1,1) = \frac{1-Q}{2} , \qquad \textrm{and} \nonumber\\
P_{XY}^Q(0,1) &= P_{XY}^Q(1,0) = \frac{Q}{2} \,. \label{eq:PXYQ}
\end{align}
We show the following, specialized bounds:
\begin{corollary}
\label{th:xi}
Let $0 < \eps < 1$ and let $0 < Q < \frac{1}{2}$. Then, for large $n$, any $\eps$-correct IR protocol satisfies
\begin{equation}
	\log |\cM| \geq \xi(n,\eps; Q)\cdot n h(Q) - \frac12 \log n - O(1),
	\label{eq:xi1} 
\end{equation}
where
\begin{equation}
	\xi(n, \eps; Q) := 1 + \frac{1}{\sqrt{n}} \, \frac{\sqrt{ v(Q) } }{ h(Q) } \Phi^{-1}(1\!-\!\eps) .   \nonumber 
\end{equation}

Here, $h(x) = - x \log x - (1-x)\log (1-x)$ and $v(x) = x(1-x) \log^2 \big( x/(1-x) \big)$. Furthermore, there exists an $\eps$-correct IR protocol with $\log |\cM| \leq \xi(n,\eps; Q)\cdot n h(Q) + \frac12 \log n + O(1)$.
\end{corollary}

The proof of Eq.~\eqref{eq:xi1} follows by specializing Theorem~\ref{th:gen} to the distribution $P_{XY}^Q$.

Moreover, numerical simulations reveal that the approximation in Corollary~\ref{th:xi} is very accurate even for small values of $n$. More precisely, we find the following exact bound:
\begin{align}
	\log |\cM| &\geq n h(Q) + \bigg( n (1-Q) - F^{-1}\Big( \eps \big( 1 + 1/\sqrt{n} \big); n, 1-Q \Big) - 1\bigg) \log \frac{1-Q}{Q}  \nonumber\\
	&\quad - \frac12 \log n - \log \frac{1}{\eps} \,,
	\label{eq:exact}
\end{align}
where $F^{-1}(\,\cdot\,; n, p)$ is the inverse of the cumulative distribution function of the binomial distribution. This bound can be evaluated numerically even for reasonably large~$n$.

\subsection{Continuous Variable QKD}

The second joint distribution of interest is the binary-Gaussian distribution:
\begin{equation}
	f_{XY}(x,y)=\frac{1}{2\sqrt{2\pi\sigma^2}}\exp\left(-\frac{(x-y)^2}{2\sigma^2}\right)\label{eq:biawgn}
\end{equation}
where $x\in\{-1,1\}$ and $y\in\mathbb R$. 

In the absence of an active eavesdropper, this distribution arises in continuous variable QKD (CVQKD) with binary modulations \cite{leverrier2009unconditional,leverrier2011continuous} and can be induced in the classical postprocessing of CVQKD with Gaussian modulation \cite{leverrier2008multidimensional,jouguet2011long}. 
For this distribution, both the conditional entropy and the conditional entropy variance do not have known closed form formulas. Abusing notation we denote them again by $h(\sigma)$ and $v(\sigma)$ respectively.
The conditional entropy is known to be \cite{leverrier2009theoretical}:
\begin{equation}
	h(\sigma) = \int_{-\infty}^\infty \phi_\sigma(y)\log(\phi_\sigma(y))dy+\frac{1}{2} \log (8\pi e \sigma^2)\label{eq:condentropybiawgn} 
\end{equation}
where
\begin{equation*}
	\phi_\sigma(y) = \frac{1}{\sqrt{8\pi\sigma^2}}\left(e^{-\frac{(y+1)^2}{2\sigma^2}}+e^{-\frac{(y-1)^2}{2\sigma^2}}\right)
\end{equation*}
The conditional entropy variance is easily found by applying Eq.~\eqref{eq:condentvar}
\begin{equation}
	v(\sigma) = e(\sigma)  - h(\sigma)^2\label{eq:varcevsigma}
\end{equation}
where
\begin{equation*}
	e(\sigma)=2 \int_{-\infty}^\infty f_{XY}(1,y)\left(\log\left(\frac{f_{XY}(1,y)}{f_{XY}(1,y)+f_{XY}(-1,y)}\right)\right)^2
\end{equation*}
These two integral forms can be solved numerically.

For this distribution, Theorem~\ref{th:gen} yields the following bound:\footnote{We here apply Theorem 1 to distributions that are continuous in $Y$. Note that the proofs leading to Theorem 1 can easily be generalized to this setting.}
\begin{corollary}
\label{th:xi2}
Let $0<\eps<1$ and let $\sigma>0$. Then, for large $n$, any $\eps$-correct IR protocol satisfies
\begin{equation}
	\log|M|\geq \xi(n,\eps;\sigma)\cdot n h(\sigma)-\frac{1}{2}\log n - O(1), \label{eq:xi2} 
\end{equation}
where 
\begin{equation*}
	\xi(n, \eps; \sigma) := 1 + \frac{1}{\sqrt{n}} \, \frac{\sqrt{ v(\sigma) } }{ h(\sigma) } \Phi^{-1}(1\!-\!\eps) .
\end{equation*}
Furthermore, there exists an $\eps$-correct IR protocol with $\log |\cM| \leq \xi(n,\eps; \sigma)\cdot n h(\sigma) + \frac12 \log n + O(1)$.
\end{corollary}

\section{Notation and Definitions}

For a finite alphabet $\cX$, we use $\cP(\cX)$ to denote the set of probability distributions on $\cX$. When $\cX$ is the real line $\cP(\cX)$ denotes the set of distributions on the Borel sets of the reals. A channel is a probabilistic kernel $W: \cX \to \cP(\cY)$ and we use $PW \in \cP(\cY)$ to denote the output distribution resulting from applying $W$ to $P \in \cP(\cX)$.
We employ the $\eps$-hypothesis testing divergence as defined in~\cite{dupuis12,tomamicheltan12}. Let $\eps \in (0,1)$ and let $P, Q \in \cP(\cZ)$. We consider binary (probabilistic) hypothesis tests $\xi : \cZ \to [0,1]$ and define the \emph{$\eps$-hypothesis testing divergence}
\begin{align*}
	D_h^{\eps}(P \| Q) := \sup \Big\{ R \in \mathbb{R} \,\Big|\, \exists\ \xi :\ \Exp_{Q} \big[\xi(Z)\big] \leq (1-\eps) e^{-R}\ \land\ \Exp_{P} \big[\xi(Z)\big] \geq 1-\eps  \Big\}.
\end{align*}

Note that $D_h^{\eps}(P \| Q)=-\log \frac{\beta_{1-\eps}(P,Q)}{1-\eps}$  where $\beta_\alpha$ is defined in Polyanskiy \emph{et al.}~\cite{polyanskiy10}. It satisfies a data-processing inequality~\cite{wang09}
\begin{align*}
	D_h^{\eps}(P \| Q) \geq D_h^{\eps}(PW \| QW) 
\end{align*}
for all channels $W$ from $\cX$ to $\cY$.

The following quantity, which characterizes the distribution of the log-likelihood ratio and is known as the \emph{divergence spectrum}~\cite{han02}, is sometimes easier to manipulate and evaluate.
\begin{align*}
	D_s^{\eps}(P \| Q) := \sup \bigg\{ R \in \mathbb{R} \,\bigg|\, \Pr_{P} \Big[ \log \frac{P}{Q} \leq R \Big] \leq \eps \bigg\} .
\end{align*}

It is intimately related to the $\eps$-hypothesis testing divergence.
For any $\delta \in (0, 1-\eps)$, we have~\cite{tomamicheltan12,tomamichel12}
\begin{align}
	\label{eq:Ds-Dh}
	D_s^{\eps}(P \| Q) - \log \frac{1}{1-\eps} \leq D_h^{\eps}(P \| Q) \leq D_s^{\eps+\delta}(P \| Q) + \log \frac{1-\eps}{\delta} .
\end{align}

For a joint probability distribution $P_{XY} \in \cP(\cX \times \cY)$, we define the Shannon conditional entropy
\begin{align*}
	H(X|Y)_P := \Exp\Big[- \log \frac{P_{XY}(X, Y)}{P_{Y}(Y)} \Big] = \sum_{x \in \cX \atop y \in \cY} P_{XY}(x,y) \left(-\log \frac{ P_{XY}(x,y) }{ P_{Y}(y) }\right) \,.
\end{align*}
and its information variance
\begin{align*}
	V(X|Y)_P :&\!\!= \Var \Big[ -\log \frac{P_{XY}(X,Y)}{P_{Y}(Y)} \Big] \nonumber\\
		&\!\!= \sum_{x \in \cX \atop y \in \cY} P_{XY}(x,y) \bigg( -\log \frac{ P_{XY}(x,y) }{ P_{Y}(y) } - H(X|Y)_P \bigg)^2.
\end{align*}
We also employ the min-entropy, which is defined as
\begin{align*}
	H_{\min}(X|Y)_P := - \log p_{\textrm{guess}}(X|Y)_P, 
\end{align*}
where $p_{\textrm{guess}}(X|Y)_P := \sum_{y \in \cY} \max_{x \in \cX} P_{XY}(x,y)$.

\section{Proofs}

\subsection{One-Shot Converse Bound for General Codes}

A general (probabilistic) one-way IR code for a finite alphabet $\cX$ is a tuple $\{ \cM, e, d \}$ consisting of a set of syndromes, $\cM$, an encoding channel $e: \cX \to \cP(\cM)$, and a decoding channel $d: \cY \times \cM \to \cP(\cX)$. We say that a code is $\eps$-correct on a joint distribution $P_{XY} \in \cP(\cX \times \cY)$ if $$\Pr_{P_{XY}}\big[X = d(Y, e(X))\big] \geq 1 - \eps .$$
The converse for probabilistic protocols clearly implies the converse for protocols where the encoder and decoder are deterministic as a special case.

We show the following one-shot lower bound on the size of the syndrome.
\begin{proposition}
\label{pr:one-shot-conv}
Any $\eps$-correct one-way IR code for $P_{XY}$ satisfies,
\begin{align*}
	\log |\cM| \geq H_{\min}(X|Y)_{Q} - D_s^{\eps+\delta}\big( P_{XY} \big\| Q_{XY} \big) \,+\, \log \delta,
\end{align*}
for any $\delta \in (0, 1-\eps)$ and any $Q_{XY} \in \cP(\cX \times \cY)$.
\end{proposition}

\begin{proof}
Let $P_{XYM\hat{X}}$ be the distribution induced by $P_{XY}$, $M \leftarrow e(X)$ and $\hat{X} \leftarrow d(Y, M)$. Analogously, $Q_{XYM\hat{X}}$ is induced by $Q_{XY} \in \cP(\cX \times \cY)$, which we fix for the remainder. We then consider the hypothesis test $\xi(X,\hat{X}) = 1\{X = \hat{X}\}$ between $P_{X\hat{X}}$ and $Q_{X\hat{X}}$. We find
\begin{align*}
	\Exp_P [\xi(X,\hat{X})] = \Pr_{P} [ X = \hat{X} ] \geq 1 - \eps
\end{align*}
and
\begin{align*}    
	\Exp_Q[\xi(X, \hat{X})] = \Pr_{Q} [ X = \hat{X} ] \leq |\cM|\, p_{\textrm{guess}}(X|Y)_{Q} .
\end{align*}

The first inequality holds by assumption that the code is $\eps$-correct. The second inequality follows from the fact that $\Pr [ X = \hat{X}] \leq p_{\textrm{guess}}(X | YM) \leq p_{\textrm{guess}}(X | Y)\, |\cM|$. 

By definition of the $\eps$-divergence and the min-entropy, we thus have
\begin{align}
	D_h^{\eps}(P_{X\hat{X}} \| Q_{X\hat{X}}) \geq H_{\min}(X|Y)_{Q}  - \log |\cM| + \log (1-\eps).
	\label{eq:p1}
\end{align}

Furthermore, Eq.~\eqref{eq:Ds-Dh} and the data-processing inequality with $d$ and $e$ yields
\begin{align*}
	D_s^{\eps+\delta}(P_{XY} \| Q_{XY}) + \log \frac{1-\eps}{\delta} &\geq D_h^{\eps}(P_{XY} \| Q_{XY}) \\
	&\geq D_h^{\eps}(P_{XYM} \| Q_{XYM}) \\
	&\geq D_h^{\eps}(P_{X\hat{X}} \| Q_{X\hat{X}}) .
\end{align*}  

Finally, the statement follows by substituting Eq.~\eqref{eq:p1} and solving for $\log |\cM|$.
\end{proof}

In the i.i.d.\ setting, it is sufficient to consider distributions of the form $Q_{XY} = U_X \times P_{Y}$, where $U_X$ is the uniform distribution on $\cX$. The bound in Prop.~\ref{pr:one-shot-conv} then simplifies to
\begin{align}
	\log |\cM| \geq \log |\cX| - D_s^{\eps+\delta}\big( P_{XY} \big\| U_X \times P_{Y} \big) \,+\, \log \delta .
	\label{eq:conv-iid}
\end{align}

However, it is unclear whether choices of $Q_{XY}$ that contain correlations between $X$ and $Y$ or are not uniform on $X$ are useful to derive tight bounds in the finite block length regime.

\subsection{Proof of Theorem~\ref{th:gen}}

The problem of information reconciliation, or source compression with side information has been studied by many authors in classical information theory. Recent work by Hayashi~\cite{hayashi08} as well as Tan and Kosut~\cite{tan12} considers the normal approximation of this problem. Here, in analogy with~\cite{tomamicheltan12}, we go one step further and also look at the logarithmic third order term.

We consider the direct and converse parts of the theorem separately.
Theorem~\ref{th:gen} then follows as an immediate corollary.
We prove slightly more precise converse and direct theorems by considering the special case where the information variance vanishes separately. Note that the bounds are tight in third order for this special case, whereas otherwise a gap of $\log n$ remains.

\begin{theorem}[Converse for IR]
\label{pr:gen}
Let $0 < \eps < 1$ and let $P_{XY}$ be a probability distribution. Any $\eps$-correct one-way IR protocol on $P_{XY}$  satisfies the following bounds:
\begin{itemize}
\item If $V(X|Y)_P > 0$, we have
\begin{align*}
	\log |\cM| \geq n H(X|Y)_{P} + \sqrt{n V(X|Y)_P}\, \Phi^{-1}(1-\eps) - \frac{1}{2} \log n - O(1),
\end{align*}
\item If $V(X|Y)_P = 0$, we have $\log |\cM| \geq n H(X|Y)_{P} + \log (1-\eps)$.
\end{itemize}
\end{theorem}

\begin{proof}
We consider an i.i.d.\ distribution $(P_{XY})^{\times n}$ and use Prop.~\ref{pr:one-shot-conv}, more precisely Eq.~\eqref{eq:conv-iid}, to get
\begin{align}
	\log |\cM| 
	&\geq n \log |\cX| -  D_s^{\eps+\delta} \big( (P_{XY})^{\times n} \big\| (U_{X} \times P_{Y})^{\times n} \big) + \log \delta \nonumber\\ 
	&= - n \sup \Bigg\{ R \in \mathbb{R} \,\Bigg|\, \Pr \Bigg[ \frac{1}{n} \sum_{i=1}^n \log \frac{P_{XY}(X_i, Y_i)}{P_{Y}(Y_i)} \leq R \Bigg] \leq \eps + \delta \Bigg\} + \log \delta \,
	\label{eq:start}
\end{align}
for any $0 < \delta < 1-\eps$. Note that we pulled $\log |\cX|$ into the information spectrum to find~\eqref{eq:start}. Next, observe that the random variables $Z_i = \log \frac{P_{XY}(X_i, Y_i)}{P_{Y}(Y_i)}$ follow an i.i.d.\ distribution, and satisfy $\Exp[Z_i] = -H(X|Y)_P$ and $\Var[Z_i] = V(X|Y)_P$.
Let us first consider the special case where $V(X|Y)_P = 0$. This implies directly that $Z_i = -H(X|Y)_P$ with probability $1$. Thus,
\begin{align*}
	\Pr \bigg[ \frac{1}{n} \sum_{i=1}^n Z_i \leq R \bigg] = \begin{cases} 0 & \textrm{if}\ R < -H(X|Y)_P \\ 1 & \textrm{if}\ R \geq -H(X|Y)_P \end{cases} \, .
\end{align*}

Hence, for any $\xi > 0$ and $\delta = 1 - \eps - \xi$, we find $\log |\cM| \geq nH(X|Y)_P + \log (1- \eps - \xi)$, proving the result in the limit $\xi \to 0$.

In the following, we may therefore assume that $V(X|Y)_P > 0$, which allows for a simple application of the Berry-Esseen theorem, which states that
\begin{align*}
	\forall R \in \mathbb{R}:\left| \Pr \bigg[ \frac{1}{n} \sum_{i=1}^n Z_i \leq R \bigg] - \Phi\left( \sqrt{n}\, \frac{ R + H(X|Y)_P }{\sqrt{V(X|Y)_P}} \right) \right| \leq 
	\frac{B}{\sqrt{n}} \,, 
\end{align*}
where 
\begin{equation*}
B:=B_0 \frac{T(X|Y)_P}{\big( \sqrt{V(X|Y)_P} \big)^{3}}
\end{equation*}
and $B_0 \leq \frac12$ is a the Berry-Esseen constant~\cite{tyurin10} and $T(X|Y)_P := \Exp \Big[ \big| \log \frac{P_Y}{P_{XY}} - H(X|Y)_P \big|^3 \Big] < \infty$ is the third moment of the information spectrum. Since $0 < B < \infty$ is finite, we find
\begin{align*}
	\log |\cM| &\geq - n \sup \Bigg\{ R \in \mathbb{R} \,\Bigg|\, \Phi\left( \sqrt{n}\, \frac{ R + H(X|Y)_P }{\sqrt{V(X|Y)_P}} \right) \leq \eps + \frac{B+1}{\sqrt{n}} \Bigg\} - \frac{1}{2} \log n \\
	&= n H(X|Y)_P \nonumber\\
	&- \sqrt{n V(X|Y)_P}\cdot \sup \bigg\{ r \in \mathbb{R} \,\bigg|\, \Phi(r) \leq \eps + \frac{B+1}{\sqrt{n}} \bigg\} - \frac12 \log n \\
	&= n H(X|Y)_P - \sqrt{n V(X|Y)_P}\ \Phi^{-1}\Big(\eps + \frac{B+1}{\sqrt{n}}\Big) - \frac{1}{2} \log n \,.
\end{align*}

Here, we chose $\delta = 1/\sqrt{n}$, implicitly assuming that $n > (B+1)^2(1 - \eps)^{-2}$ is sufficiently large. Since $\Phi^{-1}$ is continuously differentiable except at the boundaries, there exists a constant $\gamma$ such that
\begin{align*}
	\Phi^{-1}\Big(\eps + \frac{B+1}{\sqrt{n}}\Big) \leq \Phi^{-1}(\eps) + \gamma\, \frac{B+1}{\sqrt{n}} .
\end{align*}

Since $V(X|Y)_P < \infty$, this then leads to the desired bound
\begin{align}
	\log |\cM| &\geq n H(X|Y)_P - \sqrt{n V(X|Y)_P}\ \Phi^{-1}(\eps) - \frac12 \log n \nonumber\\
	&\qquad\qquad\qquad\qquad\qquad\qquad - \gamma \bigg( B_0 \frac{T(X|Y)_P}{V(X|Y)_P} + \sqrt{V(X|Y)_P} \bigg) \label{eq:thisone}.
\end{align}
\end{proof}

The constant term in~\eqref{eq:thisone} can be simplified when $\eps < \frac{1}{2}$ and $n > (B+1)^2(\frac{1}{2} - \eps)^{-2}$. We get
\begin{align*}
	\log |\cM| &\geq n H(X|Y)_P - \sqrt{n V(X|Y)_P}\ \Phi^{-1}(\eps) - \frac12 \log n \\
	&- \frac{1}{\varphi(\Phi^{-1}(\eps))} \cdot \frac{3 T(X|Y)_P}{2 V(X|Y)_P} ,
\end{align*}
where we used that $B_0 \leq \frac{1}{2}$ and $\big(\sqrt{V(X|Y)_P}\big)^3 \leq T(X|Y)_P$. Moreover, we note that the choice $\gamma = \frac{\diff \Phi^{-1}}{\diff \eps} \big|_{\eps} = \frac{1}{\varphi(\Phi^{-1}(\eps))}$ is sufficient (and also necessary for large $n$) due to concavity of $\Phi^{-1}$ on $(0, \frac{1}{2})$. Here, $\varphi(x) = \frac{\diff \Phi}{\diff x} \big|_x = \frac{1}{\sqrt{2\pi}} \exp \big(-x^2/2 \big)$ denotes the probability density function of the standard normal distribution. The constant term behaves very badly for small $\eps$, e.g., we find
\begin{align*}
	\frac{1}{\varphi \big(\Phi^{-1}\big( 10^{-4} \big) \big)} \approx 2.5 \cdot 10^3
\end{align*}
for a typical value of $\eps$. Nonetheless, the normal approximation in Theorem~\ref{pr:gen} is often very accurate.

\begin{theorem}[Achievability for IR]
Let $0 < \eps < 1$ and let $P_{XY}$ be a probability distribution. There exists an $\eps$-correct one-way IR protocol with the following property:
\begin{itemize}
\item If $V(X|Y)_P > 0$, we have
\begin{align*}
	\log |\cM| \leq n H(X|Y)_P + \sqrt{n V(X|Y)_P}\, \Phi^{-1}(1-\eps) + \frac12 \log n + O(1) .
\end{align*}
\item If $V(X|Y)_P = 0$, we have $\log |\cM| \leq n H(X|Y)_{P} - \log \eps$.
\end{itemize}
\end{theorem}

\begin{proof}
We employ a one-shot achievability bound due to~\cite{yassaee13} (we use the variant in~\cite[Cor.~12]{beigi13b}), which, for every $0 < \delta < \eps$, ensures the existence of an $\eps$-correct protocol with
\begin{align*}
	\log |\cM| \leq n \log |\cX| - D_s^{\eps-\delta}\big( (P_{XY})^{\times n} \,\big\|\, (U_{X} \times P_{Y})^{\times n}\big) - \log \delta + 1 .
\end{align*}

The remaining steps are exactly analogous to the steps taken in the proof of the converse asymptotic expansion, and we omit them here.
\end{proof}

\subsection{Proof of Corollary~\ref{th:xi}}

The corollary is a trivial specialization of Theorem~\ref{th:gen} and it only remains to evaluate $H(X|Y)_P$ and $V(X|Y)_P$ for the distribution in Eq.~\eqref{eq:PXYQ}. We find
\begin{align*}
	H(X|Y)_P &= - \sum_{x,y} P_{XY}(x,y) \log \frac{P_{XY}(x,y)}{P_{Y}(y)} \\
	&= - Q \log Q - (1-Q) \log (1-Q) =: h(Q),
\end{align*}
and
\begin{align*}
	V(X|Y)_P &= \sum_{x, y} P_{XY}(x,y)\bigg( \log \frac{P_{XY}(x,y)}{P_{Y}(y)} + h(Q) \bigg)^2 \\
	&= Q \bigg( (1-Q)\log Q - (1-Q) \log (1-Q) \bigg)^2 \\
	&\qquad+ (1-Q) \bigg( Q \log (1-Q) - Q \log Q \bigg)^2\\
	&= \big( Q (1 - Q)^2 + (1-Q) Q^2 \big) \big( \log Q - \log (1-Q) \big)^2 \\
	&= Q( 1 - Q) \Big( \log \frac{ Q }{ 1-Q} \Big)^2 =: v(Q) .
\end{align*}

\subsection{Exact Converse Bound for $(\eps,Q)$-correct Codes}

Let us state a more precise lower bound on $\log |\cM|$ that is valid for all $n$ and can be evaluated numerically for large $n$. This bound has the advantage that it does not contain unspecified contributions of the form $O(1)$. In particular, it does not suffer from the problem of potentially large constant terms as discussed above.

\begin{proposition}
\label{pr:xi-converse}
Let $0 < \eps < 1$ and let $0 < Q < \frac{1}{2}$. Then, any $(\eps, Q)$-correct one-way error correction code on a block of length $n$ satisfies
\begin{align*}
	\log |\cM| &\geq n h(Q) \\
	&+ \bigg( n (1-Q) - F^{-1}\Big( \eps \big( 1 + 1/\sqrt{n} \big); n, 1-Q \Big) - 1\bigg) \log \frac{1-Q}{Q} \\
	&- \frac12 \log n - \log \frac{1}{\eps} \,,
\end{align*} 
where $F^{-1}(\,\cdot\,; n, p)$ is the inverse of the cumulative distribution function of the binomial distribution, i.e. $F(k; n, p) := \sum_{\ell=0}^k {n \choose \ell} p^{\ell} (1-p)^{n-\ell}$ and $F^{-1}(\eps; n,p) := \max \{ k \in \mathbb{N} \,|\, F(k; n, p) \leq \eps \}$. 
\end{proposition}

\begin{proof}
We repeat Eq.~\eqref{eq:start}, where we found
\begin{align*}
	\log |\cM| &\geq - \sup \Bigg\{ R \in \mathbb{R} \,\Bigg|\, \Pr \Bigg[ \sum_{i=1}^n \underbrace{ \log \frac{P_{XX'}(X_i, X'_i)}{U_{X'}(X'_i)} }_{ =:\, Z_i } \leq R \Bigg] \leq \eps + \delta \Bigg\} + \log \delta \,.
\end{align*}
for any $0 < \delta < 1-\eps$. Here, we further used that $P_{X'}$ is uniform so that the random variables $Z_i$ are of the simple form
\begin{align*}
	\Pr_P \big[ Z_i = \log Q \big] = Q \quad \textrm{and} \quad \Pr_P \big[ Z_i = \log (1-Q) \big] = 1 - Q \,.
\end{align*}

When $Q \neq \frac{1}{2}$, we can rescale this into a Bernoulli trial: $$B_i = \big( Z_i - \log Q \big) \big( \log \frac{1-Q}{Q} \big)^{-1}\ .$$ Thus, by an appropriate change of variable, we get
\begin{align}
	\log &|M| \geq\nonumber\\ 
	&\geq - \Bigg( n \log Q + \log \frac{1-Q}{Q} \cdot \sup \bigg\{ k \in \mathbb{N} \,\bigg|\, \Pr \bigg[ \sum_{i=1}^n B_i \leq k  \bigg] \leq \eps + \delta \bigg\} \Bigg) + \log \delta \nonumber\\
	&= n h(Q) \nonumber\\
	&+ \bigg( n (1-Q) - \max \Big\{ k \in \mathbb{N}\,\Big|\, F(k -1; n, 1-Q) \leq \eps + \delta \Big\} \bigg) \log \frac{1-Q}{Q} + \log \delta \label{eq:num}\\
	&= n h(Q) + \bigg( \min \Big\{ k \in \mathbb{N}\,\Big|\, F(k; n, Q) \geq 1 - \eps - \delta \Big\} - n Q \bigg) \log \frac{1-Q}{Q} + \log \delta. \nonumber
\end{align}

The remaining optimizations over $k$ and $\delta$ can be done numerically. Alternatively, we are free to choose $\delta = \frac{\eps}{\sqrt{n}}$ in Eq.~\eqref{eq:num} to conclude the proof.
\end{proof}

\subsection{Proof of Corollary \ref{th:xi2}}

In order to prove Corollary \ref{th:xi2}, we just need to evaluate the conditional entropy and entropy variances for the binary-Gaussian distribution Eq.~\eqref{eq:biawgn}. For the sake of completeness, we do the explicit calculations. For the conditional entropy we obtain,
\begin{align}
	H(X|Y)_f &= -\int_{-\infty}^\infty dy \sum_{x\in\{-1,1\}}f_{XY}(x,y)\left(\log\frac{f_{XY}(x,y)}{f_Y(y)}\right)\nonumber\\
	&= -\int_{-\infty}^\infty dy\sum_{x\in\{-1,1\}}f_{XY}(x,y)\left(\log f_{XY}(x,y)\right)\nonumber\\
	&\quad +\int_{-\infty}^\infty dy f_{Y}(y)\log\left(f_Y(y)\right)\label{eq:entbiawgntwoterms}
\end{align}

Let us expand separately the first term in Eq.~\eqref{eq:entbiawgntwoterms}:
\begin{align}
	\int_{-\infty}^\infty &dy\sum_{x\in\{-1,1\}}f_{XY}(x,y)\left(\log f_{XY}(x,y)\right)\nonumber\\
	&= \int_{-\infty}^\infty \sum_{x\in\{-1,1\}}dy\frac{1}{\sqrt{8\pi\sigma^2}}\exp\left(-\frac{(x-y)^2}{2\sigma^2}\right)\left(\log \frac{1}{\sqrt{8\pi\sigma^2}}\exp\left(-\frac{(x-y)^2}{2\sigma^2}\right)\right)\nonumber\\
	&=\int_{-\infty}^\infty \sum_{x\in\{-1,1\}}dy\frac{1}{\sqrt{8\pi\sigma^2}}\exp\left(-\frac{(x-y)^2}{2\sigma^2}\right)\left(-\frac{1}{2}\log 8\pi\sigma^2-\frac{(x-y)^2}{2\sigma^2}\log e\right)\nonumber\\ 
	&=-\frac{1}{2}\log 8\pi\sigma^2-\frac{\log e}{2\sigma^2}\int_{-\infty}^\infty \sum_{x\in\{-1,1\}}dy\frac{1}{\sqrt{8\pi\sigma^2}}\exp\left(-\frac{(x-y)^2}{2\sigma^2}\right)\left(x-y\right)^2\nonumber\\ 
	&=-\frac{1}{2}\log 8\pi\sigma^2-\frac{\log e}{2\sigma^2}\int_{-\infty}^\infty dy\frac{1}{\sqrt{2\pi\sigma^2}}\exp\left(-\frac{y^2}{2\sigma^2}\right)y^2\nonumber\\ 
	&=-\frac{1}{2}\log 8\pi\sigma^2e\label{eq:tmp1}
\end{align}

The marginal on $Y$ can be found to be:
\begin{align}
	f_Y(y)&=\sum_{x\in\{-1,1\}}f_{XY}(x,y)\nonumber\\
	      &=\frac{1}{\sqrt{8\pi\sigma^2}}\left(\exp\left({-\frac{(y+1)^2}{2\sigma^2}}\right)+\exp\left({-\frac{(y-1)^2}{2\sigma^2}}\right)\right)\label{eq:tmp2}
\end{align}

It follows that $H(X|Y)_f=h(\sigma)$ by plugging Eq.~\eqref{eq:tmp1} and Eq.~\eqref{eq:tmp2} back into Eq.~\eqref{eq:entbiawgntwoterms}.

Now let us prove that the conditional entropy variance is given by Eq.~\eqref{eq:varcevsigma}. 
\begin{align}
	V(X|Y)_f :&= \Var \left[-\log \frac{f_{XY}}{f_{Y}}\right]\nonumber\\
	&= \Exp\left[ \left(-\log \frac{f_{XY}}{f_{Y}}\right)^2\right]- \left(\Exp\left[-\log \frac{f_{XY}}{f_{Y}}\right]\right)^2\nonumber\\
	&= \Exp\left[ \left(-\log \frac{f_{XY}}{f_{Y}}\right)^2\right]-(h(\sigma))^2\label{eq:rhsvxy}
\end{align}

We conclude by identifying the first term in the right hand side of Eq.~\eqref{eq:rhsvxy} with $e(\sigma)$:
\begin{align}
	\Exp\left[ \left(-\log \frac{f_{XY}}{f_{Y}}\right)^2\right] &=\int_{-\infty}^\infty dy\sum_{x\in\{-1,1\}}f_{XY}(x,y)\left(-\log\frac{f_{XY}(x,y)}{f_{Y}(y)} \right)^2\nonumber\\
	&=2\int_{-\infty}^\infty dyf_{XY}(1,y)\left(-\log\frac{f_{XY}(1,y)}{f_{Y}(y)} \right)^2\nonumber
\end{align}
where the last equality follows because $f_{XY}(1,y)=f_{XY}(-1,-y)$.

\section{Numerical Results}
\label{sec:results}

As shown above, $\log |\cM| \approx \xi(n, \eps; \,\cdot\,) n h(\,\cdot\,)$ is theoretically achievable for both binary-binary and binary-Gaussian distributions, and optimal up to additive constants. However, this implies that, for instance in the binary-binary case, the approximation $\log |\cM| \approx 1.1 n h(Q)$ is provably too optimistic if $\xi(n, \eps; Q) > 1.1$, e.g.\ for $n \le 10^4$, $Q \ge 2.5\%$, and $\eps = 10^{-2}$. The function $\xi(\,\cdot\,, \eps; Q)$ is plotted in Fig.~\ref{fig:xi} for different values of $\eps$ and $Q$. 

Moreover, theoretical achievability only ensures the existence of an information reconciliation (error correcting) code without actually constructing it. In fact, it is not known if efficient codes used in practical implementations can achieve the above bound. Hence, the approximation given in Corollary~\ref{th:xi} and Corollary~\ref{th:xi2} are generally too optimistic and must be checked against what can be achieved using state-of-the-art codes.

We suggest that practical information reconciliation codes for finite block lengths should be benchmarked against the fundamental limit for that block length, and not against the asymptotic limit. Moreover, we conjecture that, for some constants $\xi_{1}, \xi_{2} \geq 1$ depending only on the coding scheme used, the leaked information due to information reconciliation can be approximated well by
\begin{align}
\label{eq:fit}
	\lEC \approx \xi_1 \cdot n h(Q) + \xi_2 \cdot \sqrt{n v(Q)}\,\Phi^{-1}(1-\eps)
\end{align}
for a large range of $n$ and $Q$ ($\sigma$ for binary-Gaussian distributions) as long as $\eps$ is small enough. Here, $\xi_1$ measures how well the code achieves the asymptotic limit ($1$st order) whereas $\xi_2$ measures the $2$nd order deficiency. 

In the following we test this conjecture against some state-of-the-art error correcting codes (designed for the binary symmetric and additive white Gaussian channels, BSC and AWGN, respectively).
More precisely, we study several scenarios where we fix two of the parameters in \eqref{eq:fit} ---the failure probability $\eps$, the block length $n$, the leakage and the noise parameter--- and explore the tradeoff between the two free parameters. In each secenario, we construct codes that verify the two fixed parameters and fit $\xi_1$ and $\xi_2$ according to \eqref{eq:fit}.
For this numerical analysis we have chosen low-density parity-check (LDPC) codes following several recent implementations \cite{martinez13,pacher12,walenta13}.

We constructed two sets of LDPC codes with the progressive edge algorithm (PEG)~\cite{hu05}. We constructed the first set of codes using the following degree polynomials for the BSC:
\begin{align}
\lambda_1(x) &= 0.1560x + 0.3482x^2 + 0.1594x^{13}  + 0.3364x^{14}\nonumber\\
\lambda_2(x) &= 0.1305x + 0.2892x^2 + 0.1196x^{10} + 0.1837x^{12} + 0.2770x^{14}\nonumber\\
\lambda_3(x) &= 0.1209x + 0.2738x^2 + 0.1151x^5 + 0.2611x^{10} + 0.2291x^{14} \nonumber
\end{align}
where $\lambda_1(x)$, $\lambda_2(x)$ and $\lambda_3(x)$ were designed for coding rates $0.6$, $0.7$ and $0.8$, respectively \cite{chung01}.

And we constructed the second set of codes using these polynomials for the AWGN channel:
\begin{align}
\lambda_4(x) &= 0.16988x + 0.29342x^2 + 0.1633x^6 + 0.15835x^{11} + 0.21505x^{28} \nonumber\\
\lambda_5(x) &= 0.13372x + 0.2689x^2 + 0.00358x^6 + 0.15093x^7 + 0.01572x^8 \nonumber\\
 &+ 0.04647x^9 + 0.0001x^{10} + 0.00228x^{19} + 0.08615x^{24} + 0.02173x^{25} \nonumber\\
 &+ 0.27025x^{27} + 0.00017x^{29} \nonumber\\
\lambda_6(x) &= 0.10462x + 0.31534x^2 + 0.26969x^8 + 0.00933x^{19} + 0.02778x^{21} \nonumber\\
 &+ 0.00803x^{24} + 0.23115x^{26} + 0.03406x^{29} \nonumber
\end{align}
with code rates $0.6$, $0.7$ and $0.8$, for $\lambda_4(x)$, $\lambda_5(x)$ and $\lambda_6(x)$, respectively.

Fig.~\ref{fig:fer} and Fig.~\ref{fig:fer-awgn} show the block error rate as a function of $Q$ (the crossover probability in BSC) and SNR$=1/\sigma^2$ (the signal to noise ratio in the AWGN) for codes with rates $0.6$, $0.7$, $0.8$, and lengths $10^3$, $10^4$. The thick lines connect the simulated points while the dotted lines represent a fit following Eq. \eqref{eq:fit} (the fit values can be found in Table \ref{tab:xi}). The fit perfectly reproduces the so-called waterfall region of the codes. However, Eq. (\ref{eq:fit}) drops sharply with $Q$ for $Q\in[0,0.1]$ and with $\sigma$ for $\sigma\in[0,4]$ while LDPC codes experience an error floor. In this second region the fit can not approximate the behavior of the codes.

\begin{figure}[t]
\centering
\includegraphics[width=0.6\textwidth]{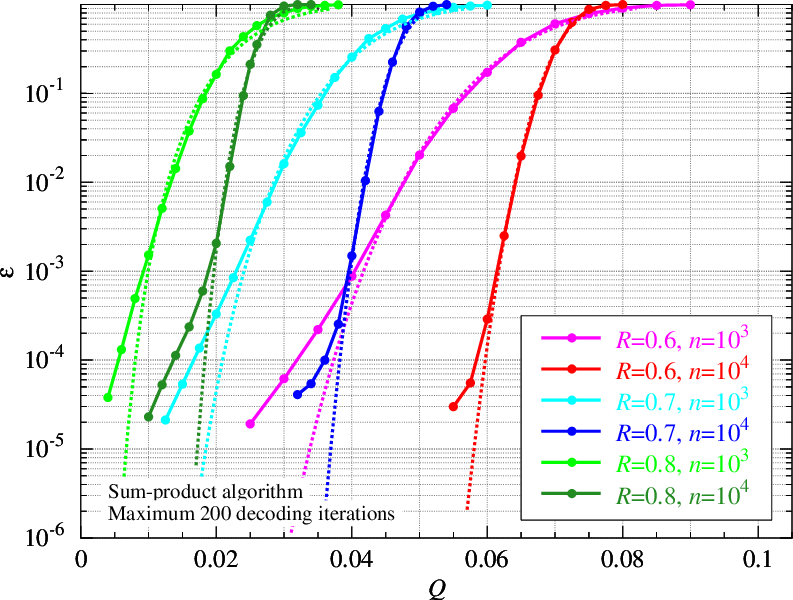}
\caption{Simulated block error rates~$\eps$ of LDPC codes of length $n=10^3$ and $n=10^4$ and code rates $R=0.6$, $R=0.7$ and $R=0.8$ as a function of quantum bit error rate $Q$.}
\label{fig:fer}
\end{figure}

\begin{figure}[t]
\centering
\includegraphics[width=0.6\textwidth]{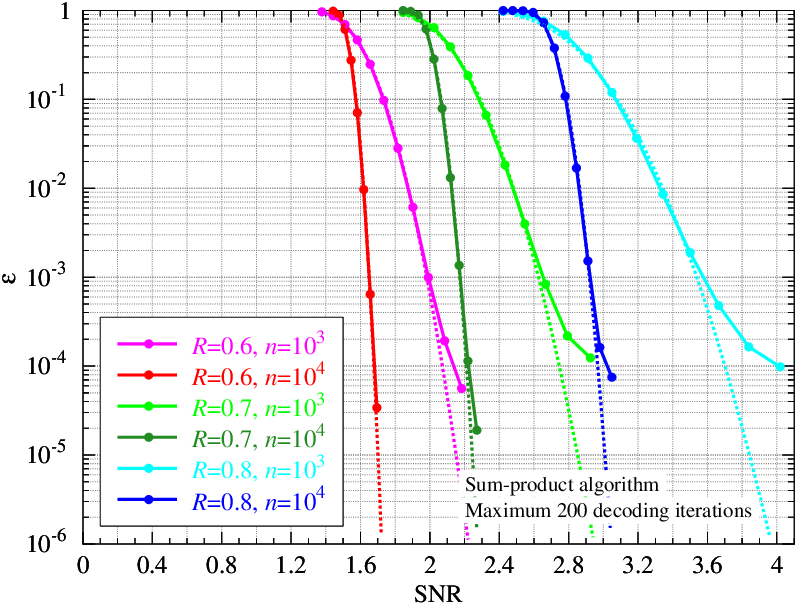}
\caption{Simulated block error rates~$\eps$ of LDPC codes of length $n=10^3$ and $n=10^4$ and code rates $R=0.6$, $R=0.7$ and $R=0.8$ as a function of SNR.}
\label{fig:fer-awgn}
\end{figure}

In Fig.~\ref{fig:xi} we plot the function $\xi(n, \eps; Q)$ and the efficiency results obtained with LDPC codes for reconciling strings following a binary-binary distribution. We chose as representative lengths $10^3$, $10^4$, $10^5$, and $10^6$. For every block length we constructed codes of rates $0.6$, $0.7$ and $0.8$ following $\lambda_1(x)$, $\lambda_2(x)$, and $\lambda_3(x)$. The points in the figure were obtained by puncturing and shortening the original codes~\cite{elkouss2011information,elkouss2012untainted} until the desired block error rate was obtained. The results show an extra inefficiency due to the use of real codes. This inefficiency shares strong similarities with the converse bound, its separation from the asymptotic value is greater for lower values of $Q$, block error rates and lengths and fades as these parameters increase. For example, for $n=10^4$, $Q=1.0\%$ and $\eps=10^{-2}$ the extra inefficiency due to the use of real codes is over $1.2$ while for $n=10^6$, $Q=5.0\%$ and $\eps=10^{-1}$ the extra inefficiency is close to $1.05$.

Similarly, in Fig.~\ref{fig:xi-awgn} we plot $\xi(n, \eps; \sigma)$ and the efficiency obtained with LDPC codes when reconciling strings following binary-Gaussian distributions. Representative lengths were also chosen $10^3$, $10^4$, and $10^5$. Codes of rates $0.6$, $0.7$, and $0.8$, following $\lambda_5(x)$, $\lambda_6(x)$ and $\lambda_7(x)$, respectively, were punctured until the desired block error rate was obtained ($\varepsilon=10^{-1}$). As in Fig.~\ref{fig:xi}, the results show an additional inefficiency due to the use of real codes. 

Finally, we address the design question posed above, that is, we study the efficiency variation as a function of the block error rate for fixed $n$ and noise parameter. We have performed this study only for the binary-binary distribution for computational reasons, but we expect similar results to hold for the binary-Gaussian. In this setting we need code constructions that allow to modulate the rate with fixed block-length.
The most natural modulating option would have been to construct codes for every $n$ of interest and augment \cite{morelos06} the codes, that is, eliminate some of the restrictions that the codewords verify. However, it is known that LDPC codes do not perform well under this rate adaptation technique \cite{varodayan06}. In consequence, we constructed a different code with the PEG algorithm for every rate. In order to obtain  a smooth efficiency curve we used the degree polynomials $\lambda_1(x)$, $\lambda_2(x)$ and $\lambda_3(x)$ for constructing all codes even with coding rates different to the design rate. 

\begin{figure}
\centering
\subfigure[$n=10^3$]{\includegraphics[width=0.6\textwidth]{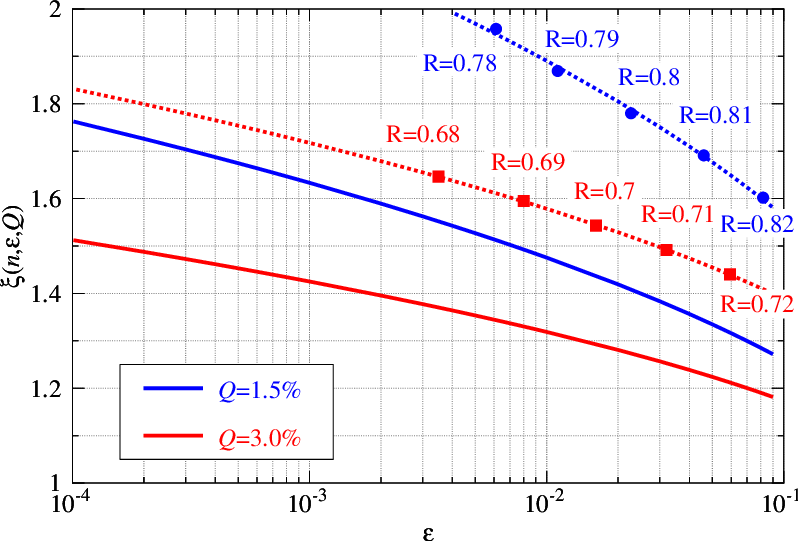}}
\subfigure[$n=10^4$.]{\includegraphics[width=0.6\textwidth]{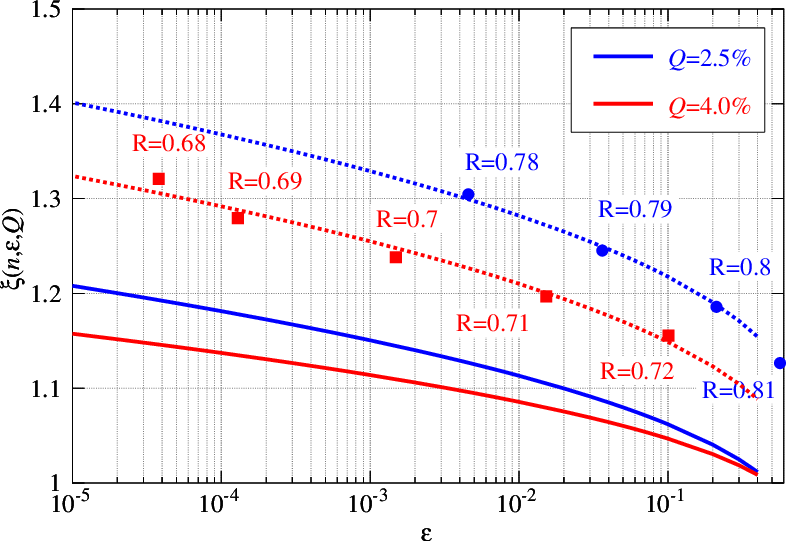}}
\caption{Ratio between the leakage and the asymptotical optimum in several scenarios as a function of the block error rate~$\eps$. Subfigures (a) and (b) show results for block lengths $10^3$ and $10^4$, respectively. In each subfigure the solid lines show the converse bound from Corollary~\ref{th:xi} while the dotted lines show the values achieved with actual LDPC codes.}
\label{fig:fervseff}
\end{figure}

Fig. \ref{fig:fervseff} shows the efficiency as a function of the block error rate. Each of the two subfigures (a) and (b) show the simulation results for codes of length $10^3$ and $10^4$, respectively. Colours blue and red correspond to $Q=1.5\%$ and $3.0\%$ in subfigure (a) and to $2.5\%$ and $4.0\%$ in subfigure (b). 
The solid lines show the bound given by Corollary~\ref{th:xi}, similar to Fig.~\ref{fig:xi} we observe that, ceteris paribus, lower values of $Q$ imply higher values of $\xi$. The points show values achieved by LDPC codes: each point represents the block error rate of a different parity check modulated code. 
Finally the dotted lines show the best least squares fit to Eq.~\ref{eq:fit}, the values of $\xi_1$ and $\xi_2$ can be found in Table~\ref{tab:xi}. From these curves we can extract some useful design information, 1) if the target failure probability is very high \cite{martinez13} then the gain obtained by increasing the block length is modest, 2) if the target failure probability is low (below $10^{-4}$) the leakage is over a fifty percent larger than the optimal one for moderate block lengths and 3) for block-length $10^5$, the largest length for which we could compute simulations in the whole block error rate region, we were unable to consistently offer efficiency values below $1.1$ and furthermore we report no point with $f$ below $1.05$.

Tables \ref{tab:xi} and \ref{tab:xi-awgn} show the values of $\xi_1$ and $\xi_2$ used in Fig.~\ref{fig:xi}, Fig.~\ref{fig:fer}, Fig.~\ref{fig:fervseff} and Fig.~\ref{fig:xi-awgn}, Fig.~\ref{fig:fer-awgn} respectively to fit the data points obtained from the simulations. In these curves $\xi_1$ is ---independently of $\eps$, $n$, $Q$, $\sigma$--- in the range $[1.05,1.16]$ while the $2$nd order deficiency $\xi_2$ is more sensible to the parameter variations. 
In the first four rows of Table \ref{tab:xi}, that correspond to Fig. \ref{fig:xi} with fixed $Q$ and $\eps$, $\xi_2$ is in the range $[2.41,3.82]$, 
for the middle six rows, that correspond to Fig. \ref{fig:fer} with fixed $n$ and leak, $\xi_2$ is in the range $[1.49,1.96]$, while for the last four rows, that correspond to Fig. \ref{fig:fervseff} with fixed $n$ and $Q$, $\xi_2$ is in the range $[1.26,1.58]$. In the first three rows of Table \ref{tab:xi-awgn}, that correspond to Fig.~\ref{fig:xi-awgn} with fixed $\sigma$ and $\eps$, $\xi_2$ is in the range $[2.58,2.71]$, while in the last six rows, that correspond to Fig. \ref{fig:fer-awgn} with fixed $n$ and leak, $\xi_2$ is in the range $[1.07,1.42]$. Note that for each scenario, the averages in these ranges could safely be used for system design purposes since necessarily codes with those $\xi_1$ and $\xi_2$ values or better exist.

\begin{table}
\centering
\begin{tabular}{cccc|cc}
$n$ & $Q$ & $\eps$ & leak &$\xi_1$ & $\xi_2$\\
\hline&&&&\\[-0.28cm]
  -          & $0.010$ & $10^{-2}$ &     -             & $1.13$ & $3.82$ \\
   -         & $0.025$ & $10^{-2}$ &     -             & $1.07$ & $3.71$\\
    -        & $0.050$ & $10^{-2}$ &     -             & $1.06$  & $3.54$ \\ 
     -       & $0.050$ & $10^{-1}$ &     -             & $1.05$ & $2.41$ \\
\hline&&&&\\[-0.28cm]
 $10^3$ & -           &  -               &  $4\cdot 10^2$   & $1.11$ & $1.39$ \\
 $10^3$ & -           &  -               &  $3\cdot 10^2$   & $1.12$ & $1.45$ \\
 $10^3$ & -           &  -               &  $2\cdot 10^2$   & $1.13$ & $1.69$\\
 $10^4$ & -           &  -               &  $4\cdot 10^3$   & $1.07$ & $1.41$\\ 
 $10^4$ & -           &  -               &  $3\cdot 10^3$   & $1.08$  & $1.44$ \\ 
 $10^4$ & -           &  -               &  $2\cdot 10^3$   & $1.11$ & $1.89$ \\
\hline&&&&\\[-0.28cm]
$10^3$ & $0.015$ &    -             &     -             & $1.16$ & $1.52$\\
$10^3$ & $0.030$ &     -            &     -             & $1.16$  & $1.31$ \\ 
$10^4$ & $0.025$ &      -           &     -             & $1.14$ & $1.26$ \\
$10^4$ & $0.040$ &       -          &     -             & $1.07$ & $1.58$\\
\hline\\
\end{tabular}
\caption{Values of $\xi_1$ and $\xi_2$ for the fitted curves in Fig.~\ref{fig:xi}, Fig.~\ref{fig:fer} and Fig.~\ref{fig:fervseff}.}
\label{tab:xi}
\end{table}

\begin{table}
\centering
\begin{tabular}{cccc|cc}
$n$ & SNR & $\eps$ & leak &$\xi_1$ & $\xi_2$\\
\hline&&&&\\[-0.28cm]
  -         & $1.6$ & $10^{-1}$ &     -             & $1.07$ & $2.58$ \\
  -         & $2.1$ & $10^{-1}$ &     -             & $1.06$ & $2.67$ \\
  -         & $2.8$ & $10^{-1}$ &     -             & $1.06$  & $2.74$ \\
\hline&&&&\\[-0.28cm]
 $10^3$ & -           &  -               &  $4\cdot 10^2$   & $1.11$ & $1.23$ \\
 $10^3$ & -           &  -               &  $3\cdot 10^2$   & $1.12$ & $1.34$ \\
 $10^3$ & -           &  -               &  $2\cdot 10^2$   & $1.13$ & $1.40$\\
 $10^4$ & -           &  -               &  $4\cdot 10^3$   & $1.08$ & $1.27$\\ 
 $10^4$ & -           &  -               &  $3\cdot 10^3$   & $1.07$  & $1.42$ \\
 $10^4$ & -           &  -               &  $2\cdot 10^3$   & $1.08$ & $1.33$ \\
\hline\\
\end{tabular}
\caption{Values of $\xi_1$ and $\xi_2$ for the fitted curves in Fig.~\ref{fig:xi-awgn} and Fig.~\ref{fig:fer-awgn}.}
\label{tab:xi-awgn}
\end{table}

\section{Conclusion}
\label{sec:conclusion}

In this paper we studied the  fundamental limits for one-way information reconciliation in the finite key regime. These limits imply that a commonly used approximation
for the information leakage during information reconciliation
is too optimistic for a range of error rates and block-lengths. We proposed a two-parameter approximation that takes into account finite key effects. 

We compared the finite length limits with LDPC codes and found a consistent range of achievable finite-length efficiencies. These efficiencies should be of use to the quantum key distribution systems designer. One question that we leave open is the study of these values for different coding families.

Finally, it is clear that PE and PA also contribute to finite-length losses in the QKD key rate. 
While it seems possible to investigate fundamental limits in PA based on the normal approximation of randomness extraction against quantum side information~\cite{tomamichel12} as a separate problem, we would in fact need to investigate it jointly with IR since there is generally a trade-off between the two tasks that needs to be optimized over.

\paragraph*{Acknowledgements}

MT thanks N.~Beaudry, S.~Bratzik, F.~Furrer, M.~Hayashi, C.C.W.~Lim, and V.Y.F.~Tan for helpful comments and pointers to related work. 
MT is supported by an Australian Research Council Discovery Early Career Researcher Award (DECRA) fellowship. 
JM has been funded by the Spanish Ministry of Economy and Competitiveness through project Continuous Variables for Quantum Communications (CVQuCo), TEC2015-70406-R. 
CP has been funded by the Vienna Science and Technology Fund (WWTF) through project ICT10-067 (HiPANQ). 
DE was supported via STW and the NWO Vidi grant ``Large quantum networks from small quantum devices".

\bibliographystyle{abbrv}
\bibliography{qinp}

\end{document}